\documentclass[conference,letterpaper]{IEEEtran}

\usepackage[utf8]{inputenc} 
\usepackage[T1]{fontenc}
\usepackage{url}
\usepackage{ifthen}
\usepackage{cite}
\usepackage[cmex10]{amsmath}

\usepackage{balance}

\usepackage{dsfont} 
\usepackage{subcaption}
\usepackage{cite}             
\usepackage[utf8]{inputenc}   
\usepackage[T1]{fontenc}     
\usepackage{url}                
\usepackage[pdftex]{graphicx} 
\usepackage[cmex10]{amsmath}  
\usepackage{amsfonts,amssymb,mathtools,bbm,verbatim}  
\usepackage{amsthm}
\usepackage{array}
\usepackage{mleftright}       
\mleftright                   
\usepackage{comment}
\usepackage{hyperref}
\hypersetup{colorlinks = true, linkcolor = blue, citecolor = blue, urlcolor=blue}
\usepackage[all]{hypcap}      
\usepackage[capitalize]{cleveref}
\usepackage{cuted}
\usepackage{caption}
     
\allowdisplaybreaks          
\usepackage{float}
\usepackage{tikz}
\usetikzlibrary{arrows.meta, positioning}
\usetikzlibrary{fit, positioning, shapes}

\crefformat{equation}{(#2#1#3)}

\def\PYX{{ P_{Y|X} }}
\def\X{{\mathcal{X}}}
\def\Y{{\mathcal{Y}}}
\def\Z{{\mathcal{Z}}}
\def\V{{\mathcal{V}}}
\def\L{{\mathcal{L}}}

\def\P{{\mathbb{P}}}
\def\E{{\mathbb{E}}}

\def\R{{\mathbb{R}}}

\def\1{{\mathbf{1}}}
\def\0{{\mathbf{0}}}
\def\I{{\mathds{1}}}

\def \taumax {{\tau_{\max}}}
\def \taumaxtwo {{\tau_{\max_{2}}}}

\def \taupair {{\tau_{ \text{pair} }}}
\def \tautrip {{\tau_{ \text{trip} }}}

\def \calP {{\mathcal{P}}}
\def \U {{\mathcal{U}}}
\def \calV {{\mathcal{V}}}
\def \calI {{\mathcal{I}}}

\renewcommand{\complement}{\mathrm{c}}

\newcommand{\T}{\mathrm{T}}

\newcommand{\pa}{\mathrm{pa}}

\newtheorem{theorem}{Theorem}

\newtheorem{proposition}{Proposition}
\newtheorem{corollary}{Corollary}

\begin{document}
\bstctlcite{IEEEexample:BSTcontrol}

\title{Bounds on Maximal Leakage\\over Bayesian Networks}

\IEEEoverridecommandlockouts

\author{
    \IEEEauthorblockN{Anuran Makur\IEEEauthorrefmark{1} and Japneet Singh\IEEEauthorrefmark{2}\thanks{The author ordering is alphabetical.} \thanks{This work was supported by the National Science Foundation (NSF) CAREER Award under Grant CCF-2337808.}}
    \IEEEauthorblockA{
        \IEEEauthorrefmark{1}Department of Computer Science and \IEEEauthorrefmark{1}\IEEEauthorrefmark{2}School of Electrical and Computer Engineering,\\
        Purdue University, West Lafayette, IN, USA\\
        Email: \{amakur, sing1041\}@purdue.edu
}}

\maketitle
\thispagestyle{plain}
\pagestyle{plain}

\begin{abstract} Maximal leakage quantifies the leakage of information from data $X \in \mathcal{X}$ due to an observation $Y$. While fundamental properties of maximal leakage, such as data processing, sub-additivity, and its connection to mutual information, are well-established, its behavior over Bayesian networks is not well-understood and existing bounds are primarily limited to binary $\mathcal{X}$. In this paper, we investigate the behavior of maximal leakage over Bayesian networks with finite alphabets. Our bounds on maximal leakage are established by utilizing coupling-based characterizations which exist for channels satisfying certain conditions. Furthermore, we provide more general conditions under which such coupling characterizations hold for $|\mathcal{X}| = 4$. In the course of our analysis, we also present a new simultaneous coupling result on maximal leakage exponents. Finally, we illustrate the effectiveness of the proposed bounds with some examples.
\end{abstract}
\section{Introduction}
Characterizing information flow over a Bayesian network is a fundamental problem in information theory. A key question in this context is: How much information leakage is incurred by a processing step about a source $X \in \X$ as it is processed over a Bayesian network? Consider a system containing sensitive data \(X\) and an adversary interacting with the system to receive a noisy observation \(Y\). The composite system between \(X\)  and \(Y\) might be complex and, in many scenarios, can be effectively modeled by a Bayesian network. An important quantity useful in quantifying the ``leakage'' of information from $X$ to any subset of nodes in the Bayesian network is maximal leakage. 

In this paper, we aim to establish bounds on maximal leakage over Bayesian networks that utilize both the structural properties of networks (such as parent node relationships) and the properties of individual channels comprising the network. Maximal leakage, introduced by \cite{issa2016}, is defined as the multiplicative increase in the probability of correctly guessing a possibly randomized function of the sensitive data $X$, given an observation $Y$. Formally, given a joint distribution \(P_{X,Y}\) on finite alphabets \(\X\) and \(\Y\), the maximal leakage from \(X\) to \(Y\) is defined as follows:
\begin{equation} \label{maximal leakage expression}
\L(X \rightarrow Y) \triangleq \sup_{U-X-Y-\hat{U}} \log \left(  \frac{\P(U = \hat{U})}{\displaystyle{\max_{u \in \U} \P(U = u)} }\right), 
\end{equation}
where \(U-X-Y-\hat{U} \) forms a Markov chain, and \(U\) and \(\hat{U}\) take values in the same finite, but arbitrary, alphabet set $\U$. It was shown in \cite{issa2016} that $\L(X \rightarrow Y)$ in the above definition reduces to Sibson's mutual information of order infinity \cite{Sibson1969,verdu2015alpha}:
\begin{equation}
\L(X \rightarrow Y) = \log\left(\sum_{y \in \Y} \max_{x \in \X : P_X(x) > 0} P_{Y|X}(y|x) \right).
\end{equation}
Maximal leakage has been applied in various domains, including security \cite{SecurityLeakage}, privacy \cite{PointwiseMaximalLeakage}, adaptive data analysis, and learning theory \cite{AdaptiveDataAnalysis}. Its fundamental properties, such as data processing, sub-additivity, and relation to mutual information \cite{issa2016}, are well understood. However, despite the basic known properties for maximal leakage, its behavior in the context of Bayesian networks remains far from fully understood.

Bayesian networks \cite{pml2Book}, which represent joint distributions through directed acyclic graphs, offer a flexible framework for modeling dependencies among a set of variables. In these networks, nodes represent random variables, and edges represent conditional dependencies. The joint probability distribution of a set of random variables $Y_1, \dots, Y_n$ in a Bayesian network is given by
\begin{equation} \label{BayesianNetworkFactorization}    
    P_{Y_1, \dots, Y_n} = \prod_{i=1}^n  P_{Y_i | \pa(Y_i)},
\end{equation}
where $\pa(Y_i)$ denotes the set of parent nodes of $Y_i$ in the network. Tools for analyzing maximal leakage in the context of Bayesian networks can be useful in extending the applications of maximal leakage to more complex scenarios that involve multiple steps of processing and iterative refinement based on data, especially in the context of adaptive data analysis \cite{DworkA,DworkB} and differential privacy \cite{WagnerBinaryCompositionTheorem}. In a simple chain-structured Bayesian network, the leakage from one node to another can be bounded by considering the conditional probabilities along the path. However, analysis of more complex networks presents a significant challenge. Notably, prior works, such as \cite{WagnerBinaryCompositionTheorem, BinaryMaximalLeakage}, have explored similar problems, but are limited to binary input settings and rely on Dobrushin-type coupling techniques. 

\subsubsection*{Main Contributions} 

In this work, we leverage recent developments in coupling characterizations of maximal leakage \cite{MakurSingh2024}\textemdash known to hold for specific channels\textemdash to derive bounds for Bayesian networks over finite alphabets. We establish these bounds by constructing a simultaneous coupling for joint probability mass functions (PMFs), which serves as a key tool in our analysis. Additionally, we also provide a more general condition under which such couplings are known to exist by extending the characterization for the special case of $|\X| = 4$. Furthermore, we provide examples illustrating some applications of our bounds.

\subsection{Related Literature}
Maximal leakage is a widely studied quantity in the study of information leakage, with applications spanning information theory \cite{issa2016}, computer security \cite{SecurityLeakage}, and machine learning \cite{GeneralizationError1}. The concept of leakage was first formalized in the context of computer security. \cite{Geoffrey2009} defined leakage as the logarithm of the multiplicative increase in the probability of correctly guessing a random variable $X$ upon observing $Y$. Building on this, \cite{SecurityLeakage} introduced a worst-case version of this metric, maximizing it over all prior distributions on $X$. A closed-form expression for maximal leakage in the discrete case was derived in \cite{issa2016}, which also provided an analysis of its several variants and their fundamental properties, and introduced a conditional version.

The concept of maximal leakage has since been generalized to various forms that unify several leakage measures across different applications. Maximal $\alpha$-leakage, introduced by \cite{AlphaLeakage}, incorporates a parameter $\alpha$ that interpolates between mutual information ($\alpha = 1$) and maximal leakage ($\alpha = \infty$). Similarly, maximal $(\alpha, \beta)$-leakage \cite{AlphaBetaLeakage} unifies several leakage measures, including local differential privacy and R\'enyi differential privacy. In privacy-related problems, pointwise maximal leakage \cite{PointwiseMaximalLeakage} quantifies leakage about a secret $X$ via individual outcomes. Other variants, such as $g$-leakage \cite{Espinoza_Min_entropy,Alvim2012}, have been proposed to address other scenarios, particularly in the context of entropy-based bounds on leakage.

Maximal leakage has been widely applied in machine learning and adaptive data analysis. \cite{AdaptiveDataAnalysis} used it to derive statistical guarantees for both adaptive and non-adaptive settings. In learning theory, maximal leakage has been used to analyze the generalization error of iterative algorithms \cite{GeneralizationError1,GeneralizationError2}. \cite{LearningCompressionLeakage} employed it to upper bound misclassification rates in meta-universal coding strategies. Additionally, \cite{HypothesisTestingMaximalLeakage} examined privacy-utility tradeoffs, using maximal leakage as a privacy metric and as a performance metric in hypothesis testing.

A closely related work, \cite{WagnerBinaryCompositionTheorem}, presents an adaptive composition theorem for the overall leakage $\L(X \rightarrow (Y_1, \dots, Y_n))$ in a Bayesian network, expressed in terms of individual leakages under the assumption of a binary input alphabet \({\X}\). Their approach employs Dobrushin-type coupling arguments and the degradation properties of composite channels; however, the analysis is confined to binary input spaces.  Furthermore, the bounds derived in this work align with established results on contraction coefficients (cf. \cite{MakurZheng2015,MakurZheng2020,Makur2019,Raginsky2016,Polyanskiy2017}) in Bayesian networks. In particular, they correspond to the bounds formulated for the contraction coefficients of Kullback-Leibler divergence and total variation distance in \cite{Polyanskiy2017}, as well as for Doeblin coefficients in \cite{MakurSingh2023a,MakurSingh2024,LuMakurSingh2024}.

\section{Background and Preliminaries}
Let $X, Y$ be random variables defined on finite alphabets $\X = [n] \triangleq \{1, \dots, n\}$ and $\Y$ respectively, and let $P_{X,Y}$ denote joint PMF of $X, Y$. Without loss of generality, we will assume that the marginal $P_{X}$ is such that $P_{X}(x) > 0 $ for all $x \in \X$. Let $P_i \in \calP_\Y$ (the probability simplex over $\Y$) denote the conditional PMF $P_{Y|X = i}$ for $i \in [n]$. Consider a collection of random variables $Y_1,\dots, Y_n \in \Y$. A \emph{coupling} of random variables $Y_1,\dots,Y_n$ distributed as $P_1,\dots,P_n$, respectively, is defined as a joint distribution $P_{Y_1,\dots,Y_n}$ that preserves the marginals. For the channel $P_{Y|X} = \big[P_1^{\T} , \dots , P_n^{\T}\big]^{\T} \in \R^{n \times |\Y|}_{\mathsf{sto}}$ formed by stacking the PMFs $P_{1}, \ldots, P_{n} \in \calP_\Y$, we introduce a quantity called the \emph{leakage exponent} denoted by $\taumax(\PYX)$. This quantity, referred to as \emph{max-Doeblin} in \cite{MakurSingh2024} is defined as 
\begin{equation}
  \begin{aligned}
      \taumax(\PYX) & \triangleq \sum_{ y \in \Y} \max\{ P_1(y), \dots, P_n(y)\}.
\end{aligned}
\end{equation}
 Under our assumption that $P_{X}(x) > 0 $ for all $x \in \X$, $\taumax(\PYX )$ corresponds to the exponentiated maximal leakage $\exp(\L(X \rightarrow Y) )$ by \cref{maximal leakage expression}. This equivalence motivates referring to $\taumax(\PYX)$ as the \textit{leakage exponent}. We will interchangeably use the notation $\taumax(\PYX ) = \taumax(P_1,\ldots,P_n)$ and similarly for related quantities.

\subsubsection*{Coupling-based Bound} 

Recently, \cite{MakurSingh2024} established a coupling-based bound for the leakage exponent \( \taumax(\PYX) \), which is formally stated in the following proposition:
\begin{proposition}[{Minimal Coupling Bound for Leakage Exponent \cite[Theorem 4]{MakurSingh2024}}] \label{prop:MaxDoeblinCoupling}  
For any random variables \( Y_1, \dots, Y_n \in \Y \) distributed according to the PMFs \( P_{1}, \dots, P_{n} \in \calP_\Y \), respectively, we have  
\begin{equation}  
\taumax(\PYX) \leq \min_{\substack{ P_{Y_1, \dots, Y_n} : \\ P_{Y_1 } = P_1, \, \dots \, , \, P_{Y_n} =P_n} } \! \sum_{y \in \Y} \P\left( \cup_{i=1}^n \{Y_i = y\}\right), \label{Eq:MaxDoeblinMinimalCoupling}  
\end{equation}  
where the minimum is taken over all couplings of \( P_{1}, \dots, P_{n} \), and \( \P(\cdot) \) denotes the probability law corresponding to the  coupling. Moreover,  if 
\begin{equation*}
\taumaxtwo(\PYX) \triangleq \sum_{y \in \Y} \max\nolimits_2\{P_{1}(y), \dots, P_{n}(y)\} \leq 1,     
\end{equation*}
 then equality in \cref{Eq:MaxDoeblinMinimalCoupling} holds, where $\max\nolimits_2\{x_1, \dots , x_n\}$ denotes the second largest value among $\{x_1, \dots , x_n\}$. 
\end{proposition}  
The bound in \cref{Eq:MaxDoeblinMinimalCoupling} follows by the  observation that for any coupling $\P(\cdot)$ we have
  \begin{align}
             \nonumber \sum_{y  \in \Y}  \P(\cup_{i=1}^{n}\{Y_{i} &=y\}) \! \geq \! \sum_{y  \in \Y} \max \{P_{1}(y), \mathbb{P}(\cup_{i=2}^{n}\{Y_{i}=y\})\}, \\
             & \geq \sum_{y  \in \Y} \max \{P_1(y), P_{2}(y), \ldots, P_{n}(y)\}, \label{lower bound for general n}
  \end{align}
  where the first inequality follows since $\mathbb{\P}(A \cup B) \geq \max \{\P(A), \P(B)\}$ and that the coupling preserves marginals, while the second inequality follows by recursively applying the same argument. When the equality in \cref{Eq:MaxDoeblinMinimalCoupling} holds, this result provides a coupling characterization for \( \taumax(\PYX) \). For \( n = 3 \), \cite{MakurSingh2024} showed that the condition \( \taumaxtwo(\PYX) \leq 1 \) is both sufficient and necessary, i.e., for \( n = 3 \) equality in \cref{Eq:MaxDoeblinMinimalCoupling} if and only if \( \taumaxtwo(\PYX) \leq 1 \). However, for \( n \geq 4 \), equality in \cref{Eq:MaxDoeblinMinimalCoupling} may hold under more general conditions. 
  
  In this work, we primarily focus on channels satisfying \( \taumaxtwo(\PYX) \leq 1 \), ensuring that equality in \cref{Eq:MaxDoeblinMinimalCoupling} is satisfied. For such channels, we leverage this coupling characterization to derive bounds on maximal leakage in Bayesian networks. Examples of such channel include the $q$-ary symmetric channel $W_{\delta}$ with crossover probability $\delta$, which satisfies \( \taumaxtwo(W_{\delta}) \leq 1 \) if \( \delta \leq 1 - 1/q \), and the erasure channel $\mathsf{E}_{\epsilon}$ with erasure probability $\epsilon$, which always satisfies \( \taumaxtwo(\mathsf{E}_{\epsilon}) \leq 1 \). Furthermore, for \( n = 4 \), we establish relaxed conditions under which equality in \cref{Eq:MaxDoeblinMinimalCoupling} holds, as detailed in \cref{thm: Minimal Coupling for leakage exponent}.

\section{Simultaneous Coupling for Maximal Leakage} \label{Sec: Simultaneous Coupling}
In this section, we establish a result on simultaneous couplings for a finite collection of joint PMFs. Specifically, we show that under the condition $\taumaxtwo(\cdot) \leq 1$ applied to one pair of marginal distributions, it is possible to construct a joint coupling that is minimal with respect to that pair of marginals. To formalize, consider $m$ joint probability distributions $P_{X_1, Y_{1}}, P_{X_{2}, Y_{2}}, \dots, P_{X_{m}, Y_{m}}$ over $\X \times \Y$, where the random variables $X_i \in \X$ and $Y_i \in \Y$ for $i \in [m]$ are defined on finite alphabets $\X$ and $\Y$. Let the marginals of the $Y_i$'s be denoted by $P_{Y_1}, P_{Y_2}, \dots, P_{Y_m}$. Then, our simultaneous coupling result is stated as follows:
\begin{theorem}[Simultaneous Coupling] \label{prop: Simultaneous Coupling}  
Given $m$ joint probability distributions $P_{X_1, Y_{1}},P_{X_{2}, Y_{2}}, \dots, P_{X_{m}, Y_{m}}$, assume that the marginals $P_{Y_{1}},P_{Y_{2}}, \dots, P_{Y_{m}}$ satisfy $\taumaxtwo(P_{Y_{1}},P_{Y_{2}}, \dots, P_{Y_{m}}) \leq 1$. Then, there exists a coupling $P_{X_{1}, Y_{1}, X_{2}, Y_{2}, \dots , X_{m}, Y_{m}}$ of $P_{X_1, Y_{1}}, \dots, P_{X_{m}, Y_{m}}$ such that
\begin{align} \label{eq: prop: Simultaneous Coupling}  
    \sum_{y \in \Y} \P\left( \cup_{i=1}^m \{Y_{i}=y\}\right)  =  \taumax\left(  P_{Y_{1}}, \dots , P_{Y_{m}}\right),
\end{align}  
where $\P(\cdot)$ denotes the probability law under the coupling.  
\end{theorem}  
The proof of this theorem is provided in \cref{Proof of prop: Simultaneous Coupling}. The construction of the simultaneous coupling utilizes the minimal coupling construction of $P_{Y_{1}}, P_{Y_{2}}, \dots, P_{Y_{m}}$ while ensuring that the marginals of the joint PMFs $P_{X_{i}, Y_{i}}$ for all $i \in [m]$ are preserved. We further note that the condition $\taumaxtwo(P_{Y_{1}},P_{Y_{2}}, \dots, P_{Y_{m}}) \leq 1$ is both sufficient and necessary for $m=3$, as discussed in the proof. For larger values, specifically $m \geq 4$, we conjecture that this condition can be further relaxed using a similar recipe for obtaining the simultaneous coupling. In particular, for $m=4$, the condition can be relaxed to a more general condition as detailed in \cref{thm: Minimal Coupling for leakage exponent} (although this condition may not be tight). 
    
\section{Bounds over Bayesian Networks}\label{Sec: Contraction over Bayesian Networks} 
In this section, we establish bounds on maximal leakage over Bayesian networks. Our objective is to establish upper bounds on the leakage exponent of the composite channel from a single source node to any sink node, based on the properties of the individual channels (such as their leakage exponents) in the Bayesian network. A Bayesian network is a probabilistic graphical model represented by a directed acyclic graph, where each vertex corresponds to a random variable defined over a finite alphabet. We assume that each variable (vertex) takes values in a finite alphabet of arbitrary length. 

Consider a Bayesian network with vertex set \(\mathsf{V}\) and a source node \(X\) taking values in \(\mathcal{X} = [n]\). Each vertex \(U\), except the source node, is associated with a conditional distribution \(P_{U|\pa(U)}\), where \(\pa(U)\) denotes the set of parents of \(U\). These conditional distributions together define the joint probability distribution over the network as in \cref{BayesianNetworkFactorization} (cf.  \cite{pml2Book}). We assume that the vertices in the network are topologically sorted. Specifically, for any node \(U\) and any non-empty subset of nodes \(V \subseteq \mathsf{V}\), we use the notation \(U > V\) to indicate that there is no directed path from \(U\) to \(V\). Let \(V \subset \mathsf{V}\) be any arbitrary non-empty subset of nodes, and suppose that the random variables associated with \(V\) take values in \(\mathcal{V}\). Let \(U\) be a sink node such that \(U > V\) in the topological ordering \cite{pml2Book}. For clarity, we include an illustration of a Bayesian network in \cref{Figure 1}.

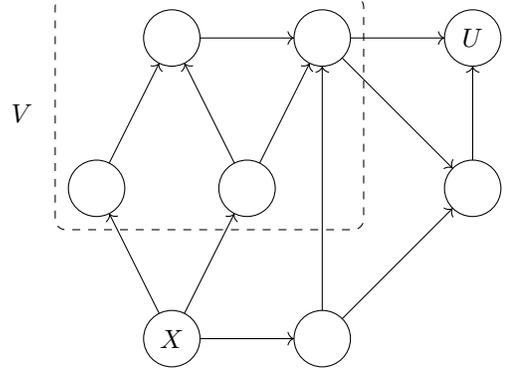
\begin{figure}
\centering
\begin{tikzpicture}[
    node distance=2cm and 2cm,
    mynode/.style={draw, circle, minimum size=0.75cm, inner sep=0pt, outer sep=0pt},
    vgroup/.style={dashed, draw, inner sep=5pt, outer sep=5pt, rounded corners},
    ]
    \node[mynode] (X) at (0,0) {$X$};
    \node[mynode] (A) at (-1,2) {};
    \node[mynode] (B) at (1,2) {};
    \node[mynode] (C) at (0,4) {};
    \node[mynode] (D) at (2,4) {};
    \node[mynode] (E) at (4,2) {};
    \node[mynode] (F) at (2,0) {};
    \node[mynode] (U) at (4,4) {$U$};
    \draw[->] (X) -- (A);
    \draw[->] (X) -- (B);
    \draw[->] (X) -- (F);
    \draw[->] (A) -- (C);
    \draw[->] (B) -- (C);
    \draw[->] (B) -- (D);
    \draw[->] (C) -- (D);
    \draw[->] (D) -- (E);
    \draw[->] (F) -- (D);
    \draw[->] (F) -- (E);
    \draw[->] (D) -- (U);
    \draw[->] (E) -- (U);
    \node[vgroup, label=left:$V$, fit=(A) (B) (C) (D)] {};
\end{tikzpicture}
\caption{Illustrative diagram of a Bayesian network. $X$ is the source node, nodes inside the dotted box belong to the set $V \subset \mathsf{V}$, and $U$ is the sink node.}
\label{Figure 1}
\end{figure}
The ensuing theorem provides bounds on the leakage exponent over a Bayesian network based on the minimal coupling and simultaneous coupling results presented in \cref{prop:MaxDoeblinCoupling} and \cref{prop: Simultaneous Coupling}. 
\begin{theorem}[Bounds on Maximal Leakage in Bayesian Networks] \label{Bounds on Maximal Leakage in Bayesian Networks}
 Let \( V \subseteq \mathsf{V} \) be a non-empty subset of nodes and let \( U \in \mathsf{V} \) be any node in $\mathsf{V}$ such that \( U > V \). Assume that the conditional distributions satisfy 
\[
\taumaxtwo(P_{U|\pa(U)}) \leq 1 \quad \text{and} \quad \taumaxtwo(P_{V|X}) \leq 1.
\]
Let the quantity \( f(P_{V \cup \pa(U) \mid X}) \) be defined as
\[
f(P_{V \cup \pa(U) \mid X}) \triangleq \sum_{v \in \mathcal{V}} \mathbb{P}\bigg( Z_1 = \dots = Z_n, \bigcup_{j=1}^n \{V_j = v\} \bigg),
\]
where $Z_i \sim P_{\pa(U)|X= i} $, $V_i \sim P_{V|X= i}$, and $\P(\cdot)$ is the probability law with respect to the simultaneous coupling of $P_{V_1,Z_1}, \dots ,P_{V_n,Z_n} $ (in \cref{prop: Simultaneous Coupling}). Then, the leakage exponent of the composite channel \( P_{V \cup \{U\} \mid X} \) satisfies the following bound:
\[
\begin{aligned}
\tau_{\max}(P_{V \cup \{U\} \mid X}) &\leq \tau_{\max}(P_{U|\pa(U)}) \cdot \tau_{\max}(P_{V \mid X}) \\
&\  - \left( \tau_{\max}(P_{U|\pa(U)}) - 1 \right) \cdot f(P_{V \cup \pa(U) \mid X}).
\end{aligned}
\]

\end{theorem}
\begin{proof}[Proof]
        Let $Z= \pa(U)$ and let its alphabet set be denoted by $\mathcal{Z}$. Let $Q_{i}$ denote the conditional distribution of $V ,Z$ given $X = {i}$, i.e., $Q_{i}=P_{V,Z \mid X={i}}$ for  $i \in [n]$. Further, we  define the random variables $\left(V_{i}, Z_{i}\right) \sim Q_{i}$ for $i \in [n]$. Since $\taumaxtwo(P_{V|X}) \leq 1$, there exists a simultaneous coupling $P_{V_1, Z_{1}, \dots ,V_{n},Z_{n}}$ of joint PMFs $Q_{1}, \dots, Q_{n}$. Let $\P(\cdot)$ denote the corresponding probability operator under the simultaneous coupling, and we have that
    \begin{equation*}
        \begin{aligned}
            \sum_{v \in \mathcal{V}} \P( \cup_{i=1}^n \{V_{i}= v\}) & =\taumax\left(P_{V \mid X}\right).
        \end{aligned}
    \end{equation*}
    Conditioned on $Z_1 = z_{1}, \ldots, Z_n = z_{n}$ for $z_1,\dots,z_n \in \Z$, define the random variables $U_{i} \sim P_{U \mid Z =z_{i}}$. We now construct a minimal coupling of conditional PMFs $P_{U \mid Z=z_{1}}, \dots , P_{U \mid Z=z_{n}}$, and let the random variables $(U_{1}, U_{2}, \ldots ,U_{n})$ be distributed according to this coupling (conditioned on $Z_1 = z_{1}, \ldots, Z_n = z_{n}$). Then, the minimal coupling of $V_{1}, Z_{1}, V_{2}, Z_{2}, \ldots, V_{n}, Z_n$ (earlier) along with the minimal coupling of $U_{1}, U_{2}, \ldots, U_{n}$ define the joint distribution $P_{V_1, U_{1}, Z_{1}, \ldots ,V_{n}, Z_{n}, U_{n}}$, which satisfies the Markov relation
    \begin{equation*}
        (V_{1}, V_{2}, \dots ,V_{n}) \longrightarrow (Z_{1}, Z_{2}, \ldots ,Z_{n}) \longrightarrow (U_{1}, U_{2}, \ldots ,U_{n}).
    \end{equation*}
    Also, by our assumption $\taumaxtwo(P_{U|Z}) \leq 1$, we have
    \begin{equation*}
    \begin{aligned}
         & \sum_{u \in \U }\mathbb{P}(\cup_{i=1}^n \{U_{i}= u\} \mid Z_{1} =z_1, \ldots, Z_{n} = z_n) -1 \\ 
         & = \taumax\left(P_{U \mid Z=z_1}, \ldots, P_{U \mid Z=z_{n}}\right) -1\\
        & \leq\left(\taumax\left(P_{U \mid Z}\right) - 1\right) \I_{\{z_1=\dots= z_n \}^\complement},
    \end{aligned}
    \end{equation*}
    where $\I_{\mathcal{A}}$ denotes the indicator function of $\mathcal{A}$, $\mathcal{B}^\complement$ denotes the complement of the set $\mathcal{B}$, and the inequality holds since the alphabet $\mathcal{Z}$ may be larger than $\left\{z_1, \ldots, z_{n}\right\}$ and the two bounds are equal when $\{z_1=\dots=z_n \}$ occurs. Now, for any fixed $v$ taking expectation conditioned on the set $\mathcal{A}_v =  \{\cup_{j=1}^n \{V_{j}= v\}\}$ on both sides, gives
    $$
    \begin{aligned}
         & \sum_{u \in \U} \E[\mathbb{P}(\cup_{i=1}^n \{U_{i}= u\} | Z_{1} =z_1, \ldots, Z_{n} = z_n)| \mathcal{A}_v ]  -1 \\
         & \leq (\taumax\left(P_{U \mid Z}\right) -1) \P(\{Z_1=\dots= Z_n \}^\complement |\mathcal{A}_v), 
    \end{aligned}
    $$
    which simplifies to
    $$
    \begin{aligned}
&        \sum_{u \in \U}  \mathbb{P}( \cup_{i=1}^n \{U_{i}= u\}| \mathcal{A}_v) - 1 \\
&      \leq (\taumax\left(P_{U \mid Z}\right) - 1) (1 - \P(\{Z_1=\dots= Z_n \} | \mathcal{A}_v )) .  
    \end{aligned}
    $$
    Now, multiplying both sides by $\mathbb{P}\left(\cup_{j=1}^n \{V_{j}= v\}\right)$ gives  
        \begin{align}
       \nonumber    &\sum_{u \in \U} \P\left( \cup_{i=1}^n \{U_{i}= u\} , \cup_{j=1}^n \{V_{j}= v\}\right)  -  \P\left(\cup_{j=1}^n \{V_{j}= v\}\right)\\ 
       \nonumber  & \leq \left(\taumax\left(P_{U | Z}\right) - 1\right) (\P(\cup_{j=1}^n \{V_{j}= v\})  - \\ 
         & \ \ \ \ \ \ \ \ \ \ \ \ \ \ \ \ \ \  \P(Z_1=\dots = Z_n, \cup_{j=1}^n \{V_{j}= v\})). \label{almost there bound}
        \end{align}
    Now observe that  
   $$ 
   \begin{aligned}
&   \sum_{v \in \V} \sum_{u \in \U} \P\left( \cup_{i=1}^n \{U_{i}= u\} , \cup_{j=1}^n \{V_{j}= v\}\right) \geq \\
&    \sum_{v \in \V} \sum_{u \in \U} \P\left( \cup_{i=1}^n \{U_{i}= u, V_{i}= v\}\right) \geq \taumax(P_{U,V |X}) 
   \end{aligned}
    $$
   Finally, taking the sum over all $v \in \calV $ on both sides of \cref{almost there bound} and utilizing the above bound and the construction of the simultaneously maximal coupling, we have 
    \begin{equation*}
    \begin{aligned}
       & \taumax\left(P_{V, U \mid X}\right) \leq \taumax\left(P_{V \mid X}\right)   + (\taumax(P_{U|Z}) -1) \times \\
       & \bigg(\taumax\left(P_{V \mid X}\right) - \sum_{v \in \V}  \P(Z_1=\dots = Z_n, \cup_{j=1}^n \{V_{j}= v\}) \bigg) . 
    \end{aligned}
    \end{equation*}
    Using the definition of $f(\cdot)$ yields the desired result. 
\end{proof}

The theorem highlights that we can iteratively ``peel off'' a node (or sets of nodes) in the Bayesian network to obtain an upper bound on the leakage exponent of the composite channel. (Note that when the graph is sparser, the number of recursive steps needed is smaller, which can imply tighter bounds.) The existence of a simultaneous coupling for the joint PMFs $P_{V_1, Z_1}, \dots, P_{V_n, Z_n} $ is central to our approach and is guaranteed by the assumption \(\taumaxtwo(P_{V|X}) \leq 1\). While the condition may appear restrictive, it arises from the lack of a general characterization when couplings achieving equality in \cref{Eq:MaxDoeblinMinimalCoupling} and \cref{eq: prop: Simultaneous Coupling} exist for a given set of PMFs. Importantly, the proof of \cref{Bounds on Maximal Leakage in Bayesian Networks} does not rely on the structure of the coupling but only on its existence and that it achieves the bound in \cref{Eq:MaxDoeblinMinimalCoupling}. Furthermore, the last term \( f(P_{V \cup \pa(U) \mid X}) \) quantifies the ``loss in leakage'' when the parents $Z_1, \dots, Z_n$ of $U$ all become equal (conditioned on $X$). We can utilize the structure of the simultaneous coupling construction presented in \cref{prop: Simultaneous Coupling} to obtain a simplified bound in terms of the Doeblin coefficient as presented below.
\begin{corollary}[Simplified Bounds on Leakage Exponent] \label{corr: Bounds on Maximal Leakage in Bayesian Networks}
    Consider the setting in \cref{Bounds on Maximal Leakage in Bayesian Networks} and assume that the conditional distributions satisfy $\taumaxtwo(P_{U|\pa(U)}) \leq 1$  and $\taumaxtwo(P_{V|X}) \leq 1$. Then, the leakage exponent of the composite channel \( P_{V \cup \{U\} \mid X} \) satisfies the following bound:
    \[
    \begin{aligned}
    \tau_{\max}( P_{V \cup \{U\} \mid X}&) \leq \tau_{\max}(P_{U|\pa(U)}) \cdot \tau_{\max}(P_{V \mid X}) \\
    &\  - \left( \tau_{\max}(P_{U|\pa(U)}) - 1 \right) \cdot \tau(P_{V \cup \pa(U) \mid X}),
    \end{aligned}
    \]
    where \( \tau(\PYX) \) is the Doeblin coefficient of the channel \( \PYX \), defined as
    \[
    \tau(\PYX) \triangleq \sum_{y \in \Y}  \min_{x \in \X} \PYX(y|x).
    \]
\end{corollary}
This corollary is a consequence of the structure of our simultaneous coupling in \cref{prop: Simultaneous Coupling}. The resulting bound depends on the Doeblin coefficient of the composite channel $\tau(P_{V \cup \pa(U) \mid X})$, which can be lower bounded by using the Doeblin coefficient of the individual channels as demonstrated in \cite[Theorem 6]{MakurSingh2024}. Consequently, our bound in \cref{corr: Bounds on Maximal Leakage in Bayesian Networks} parallels existing bounds on Bayesian networks for contraction coefficients of the Kullback-Leibler divergence and total variation distance established in  \cite{Polyanskiy2017}, as well as for Doeblin coefficients in \cite{MakurSingh2024,MakurSingh2023a}.

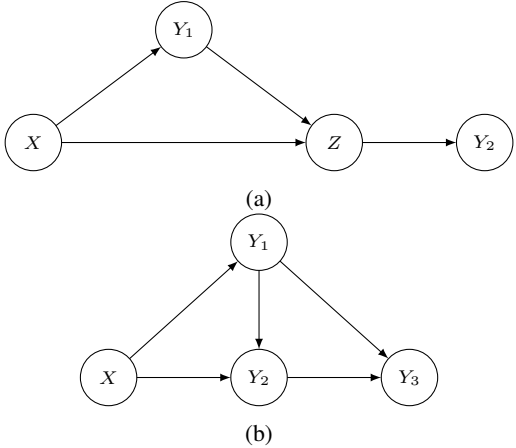
\begin{figure}[!t]
\centering

\begin{minipage}{\columnwidth}
    \centering
    \begin{tikzpicture}[
        node distance=2cm and 2cm,
        mynode/.style={draw, circle, minimum size=0.75cm, inner sep=0pt, outer sep=0pt},
        every node/.append style={font=\scriptsize},
        >={latex[length=2mm, width=1.5mm]}
    ]
        \node[mynode] (X) at (0,0) {$X$};
        \node[mynode] (Y1) at (2,1.5) {$Y_1$};
        \node[mynode] (Z) at (4,0) {$Z$};
        \node[mynode] (Y2) at (6,0) {$Y_2$};
        \draw[->] (X) -- (Y1);
        \draw[->] (Y1) -- (Z);
        \draw[->] (X) -- (Z);
        \draw[->] (Z) -- (Y2);
    \end{tikzpicture}
    \subcaption{ }
    \label{fig:example1}
\end{minipage}

\begin{minipage}{\columnwidth}
    \centering
    \begin{tikzpicture}[
        node distance=2cm and 2cm,
        mynode/.style={draw, circle, minimum size=0.75cm, inner sep=0pt, outer sep=0pt},
        every node/.append style={font=\scriptsize},
        >={latex[length=2mm, width=1.5mm]}
    ]
        \node[mynode] (X) at (0,0) {$X$};
        \node[mynode] (Y1) at (2,1.8) {$Y_1$};
        \node[mynode] (Y2) at (2,0) {$Y_2$};
        \node[mynode] (Y3) at (4,0) {$Y_3$};
        \draw[->] (X) -- (Y1);
        \draw[->] (X) -- (Y2);
        \draw[->] (Y1) -- (Y2);
        \draw[->] (Y2) -- (Y3);
        \draw[->] (Y1) -- (Y3);
    \end{tikzpicture}
    \subcaption{ }
    \label{fig:example2}
\end{minipage}
\vspace{-0.25cm}

\caption{Examples of Bayesian networks.}
\label{fig:allexamples}
\end{figure}

\subsection{Application of Bounds to Simple  Bayesian Networks}
We will now apply the bounds established in \cref{corr: Bounds on Maximal Leakage in Bayesian Networks} to different examples of Bayesian networks to illustrate their effectiveness. 
\subsubsection*{Example 1} Consider the Bayesian network shown in \cref{fig:example1}. Using \cref{corr: Bounds on Maximal Leakage in Bayesian Networks} under the assumption that $\taumaxtwo(P_{Y_2|Z} )\leq 1$ and $\taumaxtwo(P_{Y_1|X} )\leq 1$, the bound on maximal leakage is given by:
\begin{equation}
\begin{aligned}
\taumax(P_{Y_1, Y_2 \mid X}) &\leq   \taumax(P_{Y_2 \mid Z}) \taumax(P_{Y_1 \mid X}) \\
& - \big(\taumax(P_{Y_2 \mid Z}) - 1\big) \tau(P_{Y_1, Z \mid X}).
\end{aligned}
\end{equation}
Taking the logarithm on both sides, we obtain
\vspace{-1mm}
\begin{equation} \label{first example bound}
    \begin{aligned}
\L(X & \rightarrow  (Y_1, Y_2))  \leq \L(X \rightarrow Y_1) + \L(Z \rightarrow Y_2) \\
&+ \log \bigg( 1-  \frac{(\taumax(P_{Y_2 \mid Z})  - 1)  }{ \taumax(P_{Y_2 \mid Z}) } \frac{ \tau(P_{Y_1, Z \mid X} ) }{ \taumax(P_{Y_1|Z} ) } \bigg).  
\end{aligned}
\end{equation}
Observe that when the Doeblin coefficient \(\tau(P_{Y_1, Z \mid X}) >0\), the bound in \cref{first example bound} improves upon the sub-additivity property in \cite{issa2016}. Intuitively, the improvement arises because the bound accounts for the reduction in leakage due to the ``erasure of information by the Doeblin component'' \cite{GohariGunluKramer2020}, which is precisely what our proof technique for \cref{Bounds on Maximal Leakage in Bayesian Networks} aims to capture. Furthermore, this bound can be applied recursively, provided the required assumptions hold at each step.

\subsubsection*{Example 2}  
Applying \cref{corr: Bounds on Maximal Leakage in Bayesian Networks} to the Bayesian network in \cref{fig:example2}, we obtain the bound
\begin{align}
\vspace{-2mm}
\taumax(& P_{(Y_1, Y_2, Y_3)  \mid X}) \leq \taumax(P_{(Y_1,Y_2) \mid X}) \cdot \taumax(P_{Y_3 \mid (Y_1, Y_2)}) \nonumber \\ 
& -  (\taumax(P_{Y_3 \mid (Y_1, Y_2)}) - 1) \cdot \tau(P_{(Y_1,Y_2) \mid X}),       
\end{align}
under the assumption that $\taumaxtwo(P_{Y_3 \mid (Y_1, Y_2)}) \leq 1$ and $\taumaxtwo(P_{ (Y_1, Y_2)|X}) \leq 1$. Applying \cref{corr: Bounds on Maximal Leakage in Bayesian Networks}, we can further bound $\taumax(P_{(Y_1,Y_2) \mid X})$ as
\begin{equation}
    \taumax(P_{(Y_1, Y_2)|X}) \leq \taumax(P_{ Y_1 |X}) \taumax(P_{Y_2|(X, Y_1) }),
\end{equation}
 assuming $\taumaxtwo(P_{Y_1 \mid X}) \leq 1$ and $\taumaxtwo(P_{ Y_2| (X, Y_1)}) \leq 1$.

\section{Minimal Coupling Construction for Leakage Exponent for $n = 4$}\label{sec: Minimal Coupling for Leakage Exponent for n4}
In this section, we extend the known conditions beyond $\taumaxtwo(\PYX) \leq 1$ under which the coupling achieves the lower bound in \cref{Eq:MaxDoeblinMinimalCoupling} exists for the special case of $n=4$, where $n = |\X|$. Subsequently, we employ this coupling construction to directly extend the previously presented results. 
\begin{theorem}[Coupling Construction for $n=4$] \label{thm: Minimal Coupling for leakage exponent}
For $n=4$, let $P_1,P_2, P_3, P_4 \in \calP_\Y$ be the given PMFs. For $i \neq j$ and $i,j \in [4]$, let $P_{\min}(y) = \min\{ P_1(y),P_2(y),P_3(y),P_4(y) \}$ and define $N_{ij}$ as
\begin{align*}
& N_{ij} \triangleq \sum_{y \in \mathcal{Y}} \bigg( \min\{ P_i(y), P_j(y) \} - \nonumber \\ 
& \ \ \ \ \ \ \       \sum_{k \in [4] \setminus \{i,j\}}\!\!\!\! \min\{ P_i(y), P_j(y), P_k(y) \} + P_{\min}(y) \bigg).
\end{align*}
Then, there exists a coupling for $n=4$ that achieves the lower bound in \cref{Eq:MaxDoeblinMinimalCoupling}, provided the following condition holds:
\begin{align} \nonumber
\min\{ N_{12}, N_{34} \} + \min\{& N_{13}, N_{24} \}  + \min\{ N_{14}, N_{23} \} \\
& \geq \taumaxtwo(P_1,P_2,P_3,P_4) -1. \label{eq:final_condition}
\end{align}
\end{theorem}
The construction of the coupling is presented in \cref{Proof of thm: Minimal Coupling for leakage exponent}. Notably, the quantities  $N_{ij}$ are non-negative, and therefore, the condition  \cref{eq:final_condition} encompasses the condition $\taumaxtwo(\cdot) \leq 1$.  
The condition  \cref{eq:final_condition} arises after solving a system of linear equations while at the same time ensuring the non-negative of the solution. The system of linear equations has been constructed such that the coupling preserves the marginals while at the same time achieving the lower bound in \cref{Eq:MaxDoeblinMinimalCoupling}. We remark that there indeed exist PMFs that satisfy the condition in \cref{eq:final_condition}. Moreover,  we can construct a simultaneous coupling of $P_{X_1, Y_1}, \dots, P_{X_4,Y_4}$ provided that marginals $P_{Y_1}, \dots, P_{Y_4}$ satisfy the condition in \cref{eq:final_condition}. This allows us to directly extend the results in \cref{prop: Simultaneous Coupling,Bounds on Maximal Leakage in Bayesian Networks} to beyond $\taumaxtwo(\PYX) \leq 1$ for $n=4$.

\section{Conclusion and Future Directions}  
In this work, we investigated the behavior of maximal leakage over Bayesian networks on finite alphabets, establishing bounds using minimal coupling characterizations of leakage exponents and a new simultaneous coupling result. We also provided a more general condition under which the coupling characterizations held for \(|\mathcal{X}| = 4\) and demonstrated the effectiveness of our bounds through examples.  

Our work also opens up several avenues for future research. With some modifications, we believe that our framework can be generalized to related notions of maximal leakage, such as pointwise maximal leakage \cite{PointwiseMaximalLeakage}, as our coupling constructions achieve equality in \cref{lower bound for general n} for every $y \in \Y$.  Secondly, relaxing the condition \(\taumaxtwo(P_{Y|X}) \leq 1\) could lead to broader applicability of our results. Moreover, characterizing channels for which a minimal coupling achieves the bound in \cref{Eq:MaxDoeblinMinimalCoupling} remains an open problem. Finally, we believe that bounding the gap between the two terms in \cref{Eq:MaxDoeblinMinimalCoupling} could enable the extension of our bounds to general channels.

\appendices
    \section{Proof of \cref{prop: Simultaneous Coupling}} \label[appendix]{Proof of prop: Simultaneous Coupling}
        In this section, we will outline the construction of the coupling that achieves the bounds mentioned. 
        Before doing so, we introduce the constants $c_{XY}$ and $c_Y$, which are defined as follows:
        \[\begin{aligned}
            c_{XY} & \triangleq \sum_{x \in \X,y\in\Y} \min \{P_{X_1,Y_1}(x,y), \dots, P_{X_m,Y_m}(x,y)\} , \\
            c_Y & \triangleq   \sum_{y \in\Y} \min\{P_{Y_1}(y),\dots,P_{Y_m}(y)\}. 
        \end{aligned}
        \] 
        Observe that $c_Y \geq c_{XY}$. 
        Let $x^m \triangleq (x_1, x_2, \dots, x_m)$ and $y^m \triangleq (y_1, y_2, \dots, y_m)$. Additionally, we define $P_{\min}(x,y) \triangleq  \min \{P_{X_1,Y_1}(x,y), \dots, P_{X_m,Y_m}(x,y)\} $ and  
        $P_{Y_{\min}}(y) \triangleq \min \{P_{Y_1}(y), \dots, P_{Y_m}(y)\}$. We now introduce the following distributions over the product space $\X^m \times \Y^m$:
        $$
            \begin{aligned}
                G_1(x^m,y^m) & \triangleq  \frac{P_{\min}(x,y)}{c_{XY}} \I_{x^m = x, y^m = y} \\
                G_{2}(x^m,y^m) &\triangleq  \frac{P_{Y_{\min} }(y)  - \sum_{x \in \X} P_{\min}(x,y)}{c_Y -c_{XY}}   \\
                & \ \ \   \times \prod_{i=1}^m \frac{P_{X_i,Y_i}(x_i,y) -P_{\min}(x_i,y) }{P_{Y_i}(y) - \sum_{x\in\X}P_{\min}(x,y) } \I_{y^m = y}, \\
                G_3(x^m,y^m) & \triangleq \prod_{i=1}^m \frac{ P_{X_i,Y_i}(x_i,y_i) - P_{\min}(x_i,y_i) }{P_{Y_i}(y_i) - \sum_{x \in \X}P_{\min}(x,y_i) }  
                \times H(y^m) ,
            \end{aligned}
        $$
        where $H(y^m)$ is defined below and is based on the construction of minimal coupling $P_{Y_1, \dots, Y_m}(y_1, \dots, y_m)$ of the marginals $P_{Y_1}, \dots, P_{Y_m}$ as in \cite[Theorem 4]{MakurSingh2024} as
        $$
            H(y^m) \triangleq \frac{ P_{Y_1, \dots, Y_m}(y^m) - P_{Y_{\min}}(y) \I_{y^m = y }}{1  - c_Y}.
        $$
        The coupling $P_{Y_1, \dots, Y_m}$ exists when $\taumaxtwo(P_{Y_1}, \dots, P_{Y_m} ) \leq 1$ and by construction $H(y^m)$ is clearly non-negative. Also, note that $ H(y^m)$ is, in fact, a PMF, and so does $G_1, G_2, G_3$. We define the overall coupling $P_{X_1, Y_1, \dots, X_m, Y_m} $ as a convex combination of the the PMFs as:
        \[
        \begin{aligned}
                   &P_{X_1, Y_1, \dots, X_m, Y_m}(x^m, y^m) = c_{XY} G_1(x^m, y^m)  + \\
                & \ \ \ \ \ \   (c_{Y} - c_{XY}) G_2(x^m, y^m) + (1 - c_{Y}) G_3(x^m, y^m).
        \end{aligned}
        \] 
        Now, we will show that the above joint distribution has the right marginals jointly with respect to $X$ and $Y$. Define $x^{(i)} \triangleq \left(x_{1}, \ldots, x_{i-1}, x_{i+1}, \ldots, x_{m}\right)$ and $y^{(i)} \triangleq \left(y_{1}, \ldots, y_{i-1}, y_{i+1}, \ldots, y_{n}\right)$. Observe that for all $x, y$ such that $x_{i}=x $ and $y_{i}=y$, we have
        \begin{align*}
    &          \sum_{x^{(i)},y^{(i)}} P_{X_1,Y_1, \ldots, X_{m}, Y_{m}}(x^{m},  y^{m})
            =\sum_{x^{(i)}, y^{(i)}}  c_{X Y} G_{1}(x^m, y^m) \\
            & \ \ \ \ \ \  + \sum_{x^{(i)}} \sum_{y(i)}\left(c_{Y}-c_{X Y}\right) G_{2}\left(x^m, y^{m}\right) \\
            & \ \ \ \ \ \ \ \ \ \ \ \ \ \ \ \  +\sum_{x^{(i)}} \sum_{y^{(i)}}\left(1-c_{Y}\right) G_{3}\left(x^{m}, y^{m}\right) \\
            & =P_{\min }(x, y)+\left(P_{Y_{\min} }(y) -\sum_{x \in \X} P_{\min }(x, y)\right) \\
            & \ \ \ \ \ \ \ \ \   \times \left( \frac{P_{X_{i}, Y_{i}}(x, y)-P_{\min }(x, y)}{P_{Y_{i}}(y)-\sum_{x \in \X} P_{\min }(x, y)} \right)  \\ 
            & \    +\frac{\left(P_{X_{i}, Y_{i}}(x, y)-P_{\min }(x, y)\right)\left(P_{Y_{i}}(y)-P_{Y_{\min} }(y)\right)}{P_{Y_{i}}(y)-\sum_{x \in \X} P_{\min }(x, y)} \\
            & =P_{\min }(x, y)+\frac{(P_{X_{i}, Y_{i}}(x, y)-P_{\min}(x, y)) }{P_{Y_{i}}(y)-\sum_{x \in \X} P_{\min }(x, y)} \\
            & \ \ \ \ \ \ \ \ \ \ \ \ \ \ \ \ \  \times  (P_{Y_{i}}(y)-\sum_{x \in \X} P_{\min}(x, y)) \\
            & =P_{\min }(x, y)+P_{X_{i}, Y_{i}}(x, y)-P_{\min }(x, y) \\
            & =P_{X_{i}, Y_{i}}(x, y).
            \end{align*}
        Now, we will complete the proof by showing that the above coupling satisfies the condition
        \begin{equation} \label{condition to be proven}
            \sum_{y \in \Y } \P( \cup_{i=1}^m \{ Y_{i}=y\})  =\taumax(P_{Y_{1}}, \ldots, P_{Y_{n}}).
        \end{equation}
        This follows directly by observing that marginalizing the joint coupling $\P(x^m,y^m)$ with respect to $x^m$ as 
         \begin{equation*}
        \begin{aligned}
            & \sum_{x^m \in \X^m} \P(x^m,y^m)  = \sum_{x \in \X} P_{\min }(x, y) \I_{y^m = y} \\
            & \ \ \   + (P_{Y_{\min} }(y)  - \sum_{x \in \X} P_{\min}(x,y) )\I_{y^m = y}  + (1 - c_Y) H(y^m)\\ 
            & = P_{Y_{\min} }(y) \I_{y^m = y } +   P_{Y_1, \dots, Y_m}(y^m) - P_{Y_{\min}}(y) \I_{y^m = y } \\
            & = P_{Y_1, \dots, Y_m}(y^m)
        \end{aligned}
        \end{equation*}    
        Now, since the $P_{Y_1, \dots, Y_m}(y^m)$ is constructed based on \cite[Theorem 4]{MakurSingh2024}, therefore it automatically satisfies the condition in \Cref{condition to be proven}. This completes the proof. \hfill $\square$  

We remark that the condition $\taumaxtwo(P_{Y_{1}},P_{Y_{2}}, \dots, P_{Y_{m}}) \leq 1$ is necessary for $m=3$, since for any coupling of PMFs $P_1, P_2, P_3$ corresponding to random variables $Y_1, Y_2, Y_3$, we have the lower bound:
$$
\begin{aligned}
 \sum_{y  \in \Y}  \P(\cup_{i=1}^{3}\{Y_{i} =y\}) & \geq \! \sum_{y  \in \Y} \max \{P_{1}(y), \mathbb{P}(\cup_{i=2}^{3}\{Y_{i}=y\})\}  \\
             & \geq \sum_{y  \in \Y} \max \{P_1(y), P_{2}(y),  P_{3}(y)\}.    
\end{aligned}
$$
Since by \cite{MakurSingh2024}, for $n=3$ the lower bound is achieved if and only if $\taumaxtwo(P_{{1}},P_{{2}}, \dots, P_{{3}}) \leq 1$. The same argument also holds for simultaneous coupling of $P_{X_1, Y_1}, P_{X_2, Y_2}, P_{X_3, Y_3}$ as the coupling needs to preserve the marginals for each of the pairs $X_i$ and $Y_i$ for $i \in [3]$. 

    \section{Proof of \cref{thm: Minimal Coupling for leakage exponent}} \label[appendix]{Proof of thm: Minimal Coupling for leakage exponent}
    We will now demonstrate a construction of coupling that achieves the lower bound $\taumax(\PYX)$ in \cref{prop:MaxDoeblinCoupling} and for $\taumaxtwo(P_1, P_2, P_3, P_4) \geq 1$. For any subset $\calI \subset [4]$ with  $|\calI| \in \{2,3\}$, we begin by defining the following quantities:
         \begin{align}
            & \nonumber \tau  \triangleq \tau(P_1,P_2,P_3, P_4), \ \ \ \taumax  \triangleq \tau(P_1, P_2, P_3, P_4), \\
            & \nonumber \taumaxtwo  \triangleq \taumaxtwo(P_1, P_2, P_3, P_4), \\
            & \nonumber \tau_{I}  \triangleq \sum_{y\in \Y}\min_{i \in \calI} P_i(y) , \ \ \  
            \taupair  \triangleq \sum_{\calI: \calI \subset [4], |\calI| =2} \tau_{\calI}, \\
            & \tautrip  \triangleq \sum_{\calI: \calI \subset [4], |\calI| =3} \tau_{\calI}. \label{eq:tautrip}
        \end{align}
         We will utilize the following relations obtained from the maximum-minimum principle \cite{maximumminimumidentity}: 
        \begin{equation}
        \begin{aligned}
            \taumaxtwo   &=  \taupair  -2\tautrip  + 3 \tau.
        \end{aligned}
        \end{equation}
        Furthermore, we define $y^4 = (y_1, y_2, y_3, y_4)$ and $P_{\min}(y) \triangleq \min \{P_{1}(y), P_{2}(y), P_{3}(y),P_4(y)\}$. Additionally, we introduce the following PMFs:
        \begin{equation}        \label{bunch of eqns}
        \begin{aligned}
        & R_1(y)  \triangleq
        \frac{P_{1}(y)-  \min\{P_1(y), \max\{P_2(y),P_3(y),P_4(y) \} \}  }{ 1 -\tau_{12}-\tau_{13} -\tau_{14} + \tau_{123}+\tau_{124} + \tau_{134}  - \tau }, \\
        & R_2(y)  \triangleq
        \frac{P_{2}(y)-  \min\{P_2(y), \max\{P_1(y),P_3(y),P_4(y) \} \}  }{ 1 -\tau_{12}-\tau_{23} -\tau_{24} + \tau_{123}+\tau_{124} + \tau_{234}  - \tau }, \\
        & R_3(y)  \triangleq
        \frac{P_{3}(y)-  \min\{P_3(y), \max\{P_1(y),P_2(y),P_4(y) \} \}  }{ 1 -\tau_{12}-\tau_{23} -\tau_{24} + \tau_{123}+\tau_{124} + \tau_{234}  - \tau }, \\
        & R_4(y)  \triangleq
        \frac{P_{4}(y)-  \min\{P_4(y), \max\{P_1(y),P_2(y),P_3(y) \} \}  }{ 1 -\tau_{14}-\tau_{24} -\tau_{34} + \tau_{124}+\tau_{134} + \tau_{234}  - \tau }. \\
        \end{aligned}
        \end{equation}
        We now define the following PMFs where $Q_{0}$ corresponds to a PMF whose support is on the points where $y_1, y_2, y_3, y_4$ are all equal and $Q_{1,i}$ for $i \in [4]$ are the PMFs supported on the points where one of the $y_i$ is different from the remaining three $y_j$'s as:
    \begin{align*}
        & Q_{0}(y^4)  \triangleq \frac{P_{\min}(y) }{\tau}  \I_{ y_{1} = y_{2} = y_{3} = y_{4} = y }, \\
        & Q_{1,1}(y^4 ) \triangleq R_1(y_1) \times \\
        & \ \ \ \ \ \ \ \ \frac{\min\{P_{2}(y), P_{3}(y), P_{4}(y) \}- P_{\min}(y) }{\tau_{234}-\tau} \I_{ y_{2} = y_{3} = y_{4} = y } , \\
        & Q_{1,2}(y^4 ) \triangleq R_2(y_2) \times \\
        & \ \ \ \ \ \ \ \ \frac{\min\{P_{1}(y), P_{3}(y), P_{4}(y) \}- P_{\min}(y) }{\tau_{134}-\tau} \I_{ y_{1} = y_{3} = y_{4} = y } , \\
        & Q_{1,3}(y^4 ) \triangleq R_3(y_3) \times \\
        & \ \ \ \ \ \ \ \ \frac{\min\{P_{1}(y), P_{2}(y), P_{4}(y) \}- P_{\min}(y) }{\tau_{124}-\tau} \I_{ y_{1} = y_{2} = y_{4} = y } , \\
        & Q_{1,4}(y^4 ) \triangleq R_4(y_4) \times \\
        & \ \ \ \ \ \ \ \ \frac{\min\{P_{1}(y), P_{2}(y), P_{3}(y) \}- P_{\min}(y) }{\tau_{123}-\tau} \I_{ y_{1} = y_{2} = y_{3} = y } , \\
    \end{align*}
    where quantities such as $\tau_{234}$ corresponds to the definition of $\tau_{I}$ in \cref{eq:tautrip} with $I = \{2,3,4\}$. Moreover, for any distinct $i,j \in [4]$, we define quantities $T_{ij}$ and their normalizing constants $N_{ij}$:
    \begin{equation}
    \begin{aligned}
       T_{ij}(y) & \triangleq \min\{P_i(y), P_j(y) \} - \min\{P_i(y), P_j(y), P_k(y) \} \\
        & \ \ \ \  - \min\{P_i(y), P_j(y), P_l(y) \} + P_{\min}(y) \\
        N_{ij} & = \sum_{y \in \Y} T_{ij}(y)
        \end{aligned}
    \end{equation}
    where for each pair $(i,j)$, indices $k$ and $l$ are the remaining two elements of $[4]\setminus\{i,j\}$. We define the following coupling distribution:
    \begin{align*}
        Q_{2, 12}(y_{1}, y_{2}, y_{3}, y_{4}) & \triangleq  R_{1}(  y_{1} ) \times  R_{2}(  y_{2} )  \frac{T_{34} (y) }{ N_{34}} \I_{ y_{3} = y_{4} = y },\\
        Q_{2, 13}(y_{1}, y_{2}, y_{3}, y_{4}) & \triangleq  R_{1}(  y_{1} ) \times  R_{3}(  y_{3} )  \frac{T_{24} (y) }{ N_{24}} \I_{ y_{2} = y_{4} = y },\\
        Q_{2, 14}(y_{1}, y_{2}, y_{3}, y_{4}) & \triangleq  R_{1}(  y_{1} ) \times  R_{4}(  y_{4} )  \frac{T_{23} (y) }{ N_{23}} \I_{ y_{2} = y_{3} = y },\\
        Q_{2, 23}(y_{1}, y_{2}, y_{3}, y_{4}) & \triangleq  R_{2}(  y_{2} ) \times  R_{3}(  y_{3} )  \frac{T_{14} (y) }{ N_{14}} \I_{ y_{1} = y_{4} = y },\\
        Q_{2, 24}(y_{1}, y_{2}, y_{3}, y_{4}) & \triangleq  R_{2}(  y_{2} ) \times  R_{4}(  y_{4} )  \frac{T_{13} (y) }{ N_{13}} \I_{ y_{1} = y_{3} = y },\\
        Q_{2, 34}(y_{1}, y_{2}, y_{3}, y_{4}) & \triangleq  R_{3}(  y_{3} ) \times  R_{4}(  y_{4} )  \frac{T_{12} (y) }{ N_{12}} \I_{ y_{1} = y_{2} = y }.
        \end{align*}
    Define the coupling distribution as
        \begin{align} 
    & \nonumber            \P(y_{1}, y_{2}, y_{3}, y_{4}) = \tau Q_{0}(y^4)  +
                (\tau_{234}-\tau) Q_{1,1}(y^4) + \\
    & \nonumber            (\tau_{134}-\tau) Q_{1,2}(y^4) +
                (\tau_{124}-\tau) Q_{1,3}(y^4)  \\& \nonumber 
                + (\tau_{123}-\tau) Q_{1,4}(y^4)  + \beta_{34}  Q_{2, 12}(y^{4})  + \beta_{24}  Q_{2, 13}(y^{4}) \\& \nonumber
    + \beta_{23}  Q_{2, 14}(y^{4})
    + \beta_{14}  Q_{2, 23}(y^{4})
    + \beta_{13}  Q_{2, 24}(y^{4}) +\\& \nonumber
        \beta_{12}  Q_{2, 34}(y^{4})    + \alpha_{12} \frac{ T_{12}(y) }{ N_{12}  } \frac{ T_{34}(y^\prime ) }{ N_{34}  } \I_{ y_{1} =  y_{2} = y , y_{3} = y_{4} = y^\prime } 
    \\ & \nonumber + \alpha_{13} \frac{ T_{13}(y) }{ N_{13}  } \frac{ T_{24}(y^\prime ) }{ N_{24}  } \I_{ y_{1} =  y_{3} = y , y_{2} = y_{4} = y^\prime } \\&
    + \alpha_{14} \frac{ T_{14}(y) }{ N_{14}  } \frac{ T_{23}(y^\prime ) }{ N_{23}  } \I_{ y_{1} =  y_{4} = y , y_{2} = y_{3} = y^\prime },      \label{convex combination n4}     
        \end{align}
        where $\alpha_{ij}, \beta_{ij}$ are non-negative constants which will be found by solving a system of equations such that the coupling preserves the marginals and achieves the lower bound in \cref{Eq:MaxDoeblinMinimalCoupling}. Observe that our coupling structure is such that for any $\calI \subset [4]$ with cardinality $|\calI| \geq 3$ and for all $y \in \Y$, we have
        \begin{equation} \label{INTERSTINGPROP}
            \P( \cap_{i \in \calI} \{Y_i = y\}  ) = \min_{i \in \calI} P_i(y).
        \end{equation}
        Now, we impose the following constraints that for any $\calI \subset [4]$ we have the property in  \cref{INTERSTINGPROP} for $|\calI| = 2$ and for all $y \in \Y$. With some algebra, this leads us to the following constraints:
    \begin{equation} \label{pairwise constraints}
    \begin{aligned}
    & \alpha_{12} + \beta_{12} = N_{12}, \ \ \alpha_{12} + \beta_{34} = N_{34}, \\
    & \alpha_{13} + \beta_{13} = N_{13}, \ \ \alpha_{13} + \beta_{24} = N_{24}, \\
    & \alpha_{14} + \beta_{14} = N_{14}, \ \ \alpha_{14} + \beta_{23} = N_{23}. \\
    \end{aligned}
    \end{equation}
    To see the first constraint, note that for any $y \in \Y$, we have
    \begin{align*}
       & \sum_{y_3, y_4 \in \Y} \P(y,y,y_3,y_4) \stackrel{\zeta_1}{=}\underbrace{ P_{\min}(y)}_{Q_0(y^4)} \\ 
       & \ \ \  + \underbrace{ \min\{P_{1}(y), P_{2}(y), P_{4}(y) \} - P_{\min}(y) }_{Q_{1,3}(y^4)} \\ 
       & \ \ \ + \underbrace{ \min\{P_{1}(y), P_{2}(y), P_{3}(y) \}- P_{\min}(y) }_{Q_{1,4}(y^4)} + \underbrace{\beta_{12}\frac{T_{12}(y)}{N_{12}}}_{Q_{2,34}(y^4)} \\
       & \ \ \ + \alpha_{12} \frac{T_{12}(y)}{N_{12}} \\ 
       &\stackrel{\zeta_2}{=} \min\{P_{1}(y), P_{2}(y) \},
    \end{align*}
    where $\zeta_1$ follows from \cref{convex combination n4} and the underbraced quantities indicate which mixture component contributes that mass; for instance, the term underscored $Q_0(y^4)$ denotes the contribution of the $Q_0(\cdot)$ component, and similarly for $Q_{1,3}, Q_{1,4}, Q_{2,34}$ and $\zeta_2$ follows by using the contraint $\alpha_{12} + \beta_{12} = N_{12}$ and the defintion of $T_{12}(y)$. The rest of the constraints in \cref{pairwise constraints} follow by a similar calculation. Similarly, the following constraints yield the fact that the coupling preserves marginals
   \begin{equation}
       \label{triplet constraints}
           \begin{aligned}
    & \beta_{23} + \beta_{24} + \beta_{34} + \tau_{234} - \tau  = N_{R_1}, \\
    & \beta_{13} + \beta_{14} + \beta_{34} + \tau_{134} - \tau  = N_{R_2}, \\
    & \beta_{12} + \beta_{14} + \beta_{24} + \tau_{124} - \tau  = N_{R_3}, \\
    & \beta_{12} + \beta_{13} + \beta_{23} + \tau_{123} - \tau  = N_{R_4}, 
    \end{aligned}
   \end{equation}
    where $N_{R_i}$, for $i\in [4]$, is the normalizing constant (denominator term) in the definition of $R_i$ in \cref{bunch of eqns}. To see this, note that
    \begin{align*}
     &  \sum_{y_2,y_3, y_4 \in \Y} \P(y,y_2,y_3,y_4)  \stackrel{\zeta_1}{=} P_{\min}(y) + (\tau_{234} - \tau)R_1(y) \\
     & \ \ \ + \min\{P_{1}(y), P_{3}(y), P_{4}(y) \}  - P_{\min}(y) \\
      & \ \ \  + \min\{P_{1}(y), P_{2}(y), P_{4}(y) \} - P_{\min}(y) \\ 
      & \ \ \ +  \min\{P_{1}(y), P_{2}(y), P_{3}(y) \} - P_{\min}(y) \\ & 
 \ \ \      + \beta_{34}R_1(y) + \beta_{24} R_1(y)  + \beta_{23}R_1(y)  + \beta_{14} \frac{T_{14}(y)}{N_{14}} \\ 
 & \ \ \ + \beta_{13} \frac{T_{13}(y)}{N_{13}} + \beta_{12} \frac{T_{12}(y)}{N_{12}}  +  \alpha_{12} \frac{T_{12}(y)}{N_{12}} +  \alpha_{13} \frac{T_{13}(y)}{N_{13}} \\ 
 & \ \ \ \ \ \ \ \ \ \ \ \ \ \ \ \ \ \  + \alpha_{14} \frac{T_{14}(y)}{N_{14}}  \\ 
       &\stackrel{\zeta_2}{=}  (\beta_{34} +  \beta_{24} + \beta_{23} + \tau_{234} - \tau)R_1(y) + T_{12}(y) + T_{13}(y) \\ 
       & + T_{14}(y) - 2 P_{\min}(y)  + \min\{P_{1}(y), P_{2}(y), P_{3}(y) \} \\ 
       & + \min\{P_{1}(y), P_{2}(y), P_{4}(y) \} + \min\{P_{1}(y), P_{3}(y), P_{4}(y) \}    \\
       &\stackrel{\zeta_3}{=}  P_{1}(y)-  \min\{P_1(y), \max\{P_2(y),P_3(y),P_4(y) \} \}  \\
       & + \min\{P_1(y), P_2(y)\} + \min\{P_1(y), P_3(y)\} \\ 
      &  + \min\{P_1(y), P_4(y)\}  
        -   \min\{P_1(y), P_2(y), P_3(y)\} \\
        & -  \min\{P_1(y), P_2(y), P_4(y)\} -  \min\{P_1(y), P_3(y), P_4(y)\} \\
        & \ \ \ \ \ \ \ \ \ \ \ \ \ \ \ \ \ +   P_{\min}(y) \\ 
       &\stackrel{\zeta_4}{=}  P_1(y),
    \end{align*}
    where $\zeta_1$ follows from \cref{convex combination n4}, $\zeta_2$ follows using constraints \cref{pairwise constraints}, $\zeta_3$ follows by using the first constraint in \cref{triplet constraints} and $\zeta_4$ follows by the maximum-minimum principle \cite{maximumminimumidentity}.
    Now since we have shown that for any subset $\calI \subseteq [4]$ and for all $y \in \Y$, \cref{INTERSTINGPROP} holds, we have by maximum-minimums principle $\P( \cup_{i \in [4]} \{Y_i = y\}  ) = \max\{ P_1(y), P_2(y), P_3(y), P_4(y) \} $. 
    Writing each of the constraints in \cref{triplet constraints} in terms of $\alpha_{ij}$, we end up with the following single constraint
    $$
        \alpha_{12} + \alpha_{13} +  \alpha_{14} = \taumaxtwo - 1.
    $$
    Thus, for any arbitrary constants $a, b, c \geq 0$ such that $a + b + c = 1$ satisfying
    $$
    \begin{aligned}
    \min\{ N_{12}, N_{34} \} \geq a(\taumaxtwo - 1 ), \\
    \min\{ N_{13}, N_{24} \} \geq b(\taumaxtwo - 1 ), \\
    \min\{ N_{14}, N_{23} \} \geq c(\taumaxtwo - 1 ), 
    \end{aligned}
    $$
    then setting $\alpha_{12} = a (\taumaxtwo - 1) $, $\alpha_{13} = b (\taumaxtwo - 1) $,  $\alpha_{14} = c (\taumaxtwo - 1) $, and corresponding $\beta_{ij}$ determined from \cref{pairwise constraints} yields a valid construction of the coupling which preserves the marginals as well as achieves the lower bound in \cref{Eq:MaxDoeblinMinimalCoupling}. The following condition ensures the existence of such constants $a, b, c$:
        \begin{align}
    \nonumber    \min\{ N_{12}, N_{34} \} + \min\{ N_{13}, N_{24} \} +   \min\{& N_{14}, N_{23} \}  \geq    \\
    & \taumaxtwo -1.  \label{final condition}
    \end{align}
    This completes the proof.
    \qed

\balance
\bibliographystyle{IEEEtran}
\bibliography{references}

\end{document}